\documentclass[twoside,11pt]{article}

\usepackage{algorithm}
\usepackage{algpseudocode}
\usepackage{amsmath}
\usepackage{amsthm}
\usepackage{amssymb}
\usepackage{blindtext}
\usepackage{booktabs}
\usepackage{caption}
\usepackage{float}
\usepackage{graphicx}
\usepackage{multirow}
\usepackage{siunitx}
\usepackage{subcaption}
\usepackage{adjustbox}
\usepackage[T1]{fontenc}

\usepackage{ulem} 

\usepackage{jmlr2e}
\usepackage{xcolor} 

\newcounter{xaviercounter}

\newcounter{romaincounter}

\newcounter{yanniscounter}

\newcounter{anthocounter}

 \newcounter{thierrycounter}

\newcommand{\RE}{\mathrm{Re}}
\newcommand{\IM}{\mathrm{Im}}

\DeclareMathOperator*{\argmin}{argmin} 

\usepackage{lastpage}
\jmlrheading{Preprint}{2026}{1-\pageref{LastPage}}{02/26}{02/26}{21-0000}{Romain de Coudenhove}

\ShortHeadings{Fast Linear Reservoirs via Diagonalization}{de Coudenhove, Bendi-Ouis, Strock, Hinaut}
\firstpageno{1}

\begin{document}

\title{Fast Linear Reservoirs via Diagonalization}

\author{\name Romain de Coudenhove \email romain.de.coudenhove@ens.psl.eu \\
       \addr ENS PSL, Inria Center of Bordeaux University, LaBRI, IMN
       \AND
       \name Yannis Bendi-Ouis \email yannis.bendi-ouis@inria.fr \\
       \addr Inria Center of Bordeaux University, LaBRI, IMN
       \AND
       \name Anthony Strock \email astrock@stanford.edu \\
       \addr Department of Psychiatry \& Behavioral Sciences, Stanford University School of Medicine Stanford
       \AND
       \name Xavier Hinaut \email xavier.hinaut@inria.fr \\
       \addr Inria Center of Bordeaux University, LaBRI, IMN
       }
       
\editor{Preprint JMLR}

\maketitle

\begin{abstract}
We introduce a diagonalization-based optimization for Linear Echo State Networks (ESNs) that reduces the per-step computational complexity of reservoir state updates from quadratic to linear. By reformulating reservoir dynamics in the eigenbasis of the recurrent matrix, the recurrent update becomes a set of independent element-wise operations, eliminating the matrix multiplication. We further propose three methods to use our optimization depending on the situation: (i) Eigenbasis Weight Transformation (EWT), which preserves the dynamics of standard and trained Linear ESNs, (ii) End-to-End Eigenbasis Training (EET), which directly optimizes readout weights in the transformed space and (iii) Direct Parameter Generation (DPG), that bypasses matrix diagonalization by directly sampling eigenvalues and eigenvectors, achieving comparable performance to standard Linear ESNs. Across all experiments, both our methods preserve predictive accuracy while offering significant computational speedups, making them a replacement for standard Linear ESNs computations and training, and suggesting a shift of paradigm in linear ESN towards the direct selection of eigenvalues.
\end{abstract}


\begin{keywords}
  Reservoir Computing, Linear Echo State Networks, Diagonalization, Recurrent Neural Networks, Eigenvalue Decomposition
  
\end{keywords}



\section{Introduction}


Reservoir Computing (RC), most notably implemented as Echo State Networks (ESN) \citep{jaeger2004harnessing}, offers a compelling alternative to traditional Recurrent Neural Networks (RNNs) by keeping internal reservoir weights fixed and training only the output layer. While this approach simplifies the training procedure to a linear regression problem, the reservoir state update remains computationally expensive, scaling with quadratic complexity ($\mathcal{O}(N^2)$ where $N$ is the number of neurons). Consequently, there is a strong motivation to reduce this complexity without sacrificing performance.
While early RC research emphasized the necessity of non-linear reservoir dynamics, recent studies have demonstrated the surprising efficacy of linear reservoirs. Theoretically, it has been shown that a linear reservoir coupled with a non-linear readout (such as a polynomial function or neural network) is a universal approximator \citep{gonon2019reservoir}. Furthermore, linear dynamics are particularly well-suited for maximizing memory capacity, as analyzed in \citet{dambre2012information}. 
Several architectures have leveraged these properties. For instance, Legendre Memory Units (LMUs) utilize linear dynamics derived from Legendre polynomials to project input history onto an orthogonal basis, thereby extending memory capacity \citep{voelker2019legendre, chilkuri2021parallelizing}. 
Similarly, recent works by Peter Tino \citep{simon2025linear, verzelli2021input, tino2020dynamical} have explored the dynamical properties of linear models, clarifying the precise conditions under which linear ESNs can match or outperform their non-linear counterparts. 
Despite these advantages, the computational cost of matrix multiplication remains a bottleneck for large reservoirs. To address this, approaches such as Next Generation Reservoir Computing \citep{gauthier2021nextgenRC} have proposed replacing the explicit simulation of reservoir dynamics with static non-linear feature expansions, offering significant speedups. However, abandoning the recurrent dynamics entirely may not always be desirable as it limits the model to a fixed size window.

In this work, we present a method to reduce computational complexity while maintaining the explicit recurrence of the reservoir. We focus on optimizing Linear ESNs through matrix diagonalization, a transformation that converts costly matrix multiplications into efficient element-wise operations, reducing the complexity from $\mathcal{O}(N^2)$ to $\mathcal{O}(N)$. 
Beyond immediate computational efficiency, our investigation on linear reservoir is motivated by a fundamental advantage of linear systems: the ability to parallelize computation across the entire input sequence. While this principle of parallelizable linear state propagation aligns with recent advances in modern state-space models such as Mamba \citep{gu2024mamba}, our approach demonstrates that the specific findings of \citet{gupta2022diagonal} and \citet{smith2208simplified} regarding the efficacy of diagonal state spaces naturally extend to the Reservoir Computing paradigm. 
First, Section 2 formalizes the Linear ESN framework and highlights the computational cost of standard dynamics. Section 3 details our diagonalization-based reformulation which reduces the update complexity to $\mathcal{O}(N)$. Section 4 presents three practical implementations of this optimization, including the direct generation of spectral parameters. Finally, Section 5 validates these methods on standard benchmarks, followed by a discussion on the importance of spectral properties in Section 6.


\section{Linear Echo State Networks: Definition}

\begin{figure}[h]
    \centering
    \includegraphics[width=0.8\textwidth]{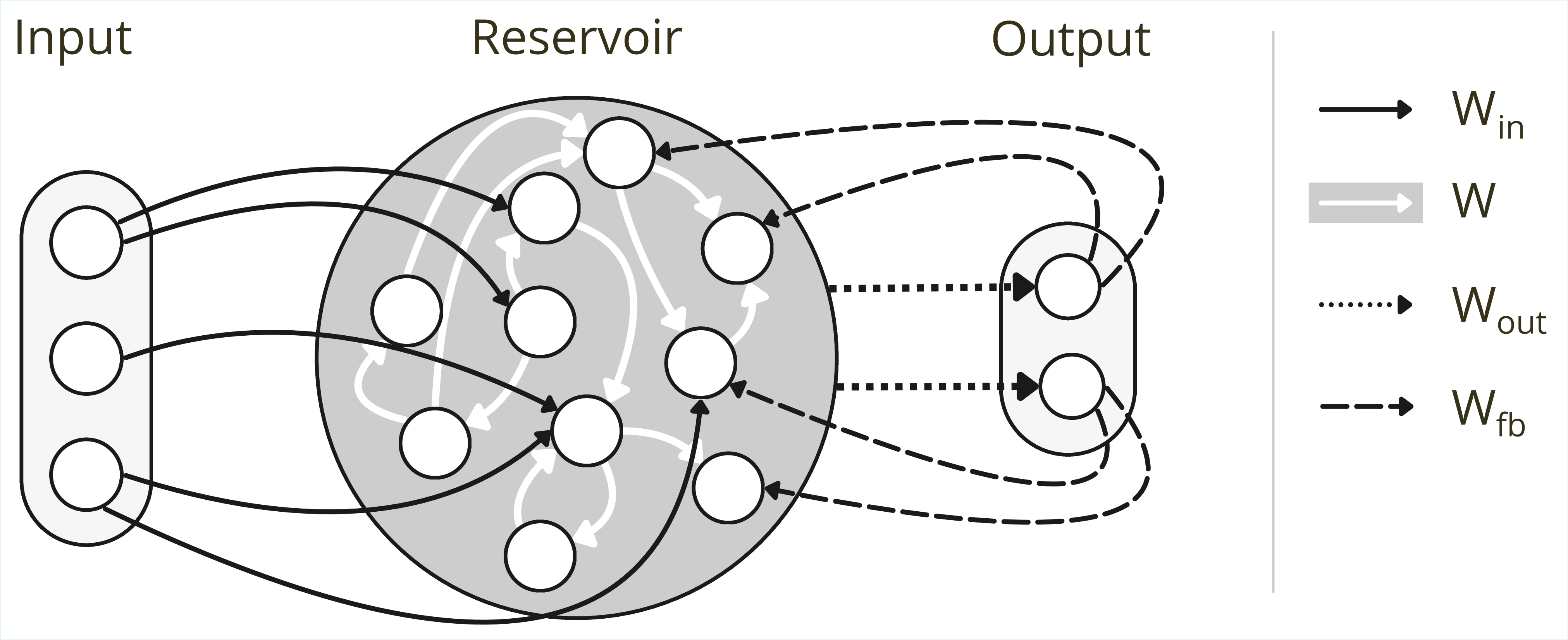}
    \caption{Standard architecture of an Echo State Network (ESN). In the Reservoir Computing paradigm, the input weights $W_{in}$, the internal recurrent weights $W$, and the optional feedback weights $W_{fb}$ are randomly initialized and kept fixed. The input is projected into a high-dimensional space within the reservoir, which utilizes its recurrent connections to maintain a deterministic state representation over time. Only the readout weights $W_{out}$ are trained to decode this internal state and produce the final output.}
    \label{fig:diag_cond}
\end{figure}

An Echo State Network (ESN) is a type of recurrent neural network particularly effective in learning complex dynamical systems. In its linear form, the network is fully characterized by a set of matrices:
\begin{itemize}
    \item $W \in \mathbb{R}^{N \times N}$, the reservoir transition matrix
    \item $W_{\text{in}} \in \mathbb{R}^{D_\text{in} \times N}$, the input matrix
    \item $W_{\text{fb}} \in \mathbb{R}^{D_\text{out} \times N}$, the feedback matrix (optional)
    \item $W_{\text{out}} \in \mathbb{R}^{N' \times D_\text{out}}$, the readout matrix 
\end{itemize}

\noindent Here, $N$ denotes the number of neurons in the reservoir, while $D_{\text{in}}$ and $D_{\text{out}}$ represent the dimensionality of the input and output signals, respectively. The dimension $N'$ corresponds to the size of the extended state vector used by the readout. Depending on the chosen configuration, this vector is formed by concatenating the reservoir state $r(t)$ with optional components: a constant bias, or the previous output $y(t-1)$. Consequently, $N'$ varies based on the inclusion of these additional terms alongside the $N$ reservoir units. \\
\noindent This structure leverages a fixed, randomly initialized reservoir $W$, where only the output weights $W_{\text{out}}$ are optimized. 
Let the time series input be $(u(t))_{t=1}^T \in \mathbb{R}^{T \times D_\text{in}}$, where $u(t) \in \mathbb{R}^{D_\text{in}}$ represents the input vector at time step $t$.
For every time step $t$, the reservoir produces an internal state $r(t) \in \mathbb{R}^N$, and the readout produces an output $y(t) \in \mathbb{R}^{D_{\text{out}}}$.

\subsection{Reservoir step}

The internal state $r(t)$ evolves according to the following dynamics. This recursive update propagates information through time via a weighted combination of past states, current inputs, and past outputs.

\begin{equation}
    r(t) = r(t-1) W + u(t) W_{\text{in}} + y(t-1) W_{\text{fb}}
    \label{eq:reservoir_step}
\end{equation}

\noindent \textbf{Note on optional input bias:} An input bias term can be added to the reservoir. However, since its effect is mathematically equivalent to adding a constant dimension of 1 to the input signal $u(t)$, we omit the input bias for the remainder of this work.

\subsection{Readout step}

The output $y(t)$ at time $t$ is a linear combination of an extended state vector (which includes a bias, the previous output, and the current reservoir state) and the readout weight matrix $W_{\text{out}}$. 

\begin{equation}
    y(t) =
    \begin{bmatrix} 1 & y(t-1) & r(t)\end{bmatrix}
    \, W_{\text{out}}, \quad \text{for } t = 1,\ldots,T
    \label{eq:readout_step}
\end{equation}

\begin{equation}
    W_\text{out} =
    \begin{bmatrix}
        W_\text{out,bias} \\
        W_\text{out,out} \\
        W_\text{out,res}
    \end{bmatrix}
\end{equation}

\noindent \textbf{Note on optional features:} Depending on the implementation, both bias and previous output are optional, only the current reservoir state is mandatory. Consequently $W_\text{out,bias}$ and $W_\text{out,out}$ are optional too. \\

\noindent \textbf{Note on initial states:} According to these equations, the model requires an initial state for the reservoir and for the output if there is a feedback. By specifying these initial conditions, the entire reservoir and output trajectories are uniquely determined for a given input sequence. Numerous methods exist to initialize those vectors, but in most cases they are initialized either by setting them to the zero vector, or by using a preliminary « warmup » phase where the network runs for a short period to stabilize the state dynamics before training starts. In this context, we set the initial conditions to zero.

\subsection{Leaking Rate}

The leaking rate $lr \in ]0,1]$ is a hyperparameter in reservoir computing that controls the temporal decay of internal states. It introduces a form of memory retention by blending the previous state and the potential new state at each time step. A smaller value of $lr$ results in a slower update and consequently a larger memory capacity, while $lr=1$ yields the fastest update and consequently the smallest memory capacity.
In practice in Linear ESNs, the leaking rate modifies the transition matrix $W$ by mixing it with the identity:  
\begin{equation}
    W^{(lr)} = lr \cdot W + (1 - lr) \cdot I
\end{equation}

\noindent The input and feedback weights are similarly rescaled:
\begin{align*}
    W_{\text{in}}^{(lr)} &= lr \cdot W_{\text{in}}\\
    W_{\text{fb}}^{(lr)} &= lr \cdot W_{\text{fb}}
\end{align*}

\noindent As a result, the full reservoir dynamics under leaking rate become:
\begin{align}
    r(0) &= 0, \\
    r(t) &= r(t-1)\, W^{(lr)} + u(t)\, W_{\text{in}}^{(lr)} + y(t)\, W_{\text{fb}}^{(lr)}
\end{align}

\noindent Leaking rate only reparametrizes the original matrices $W$, $W_{\text{in}}$ and $W_{\text{fb}}$, which does not alter the structure of the network. Thus, the exact same optimization method (section \ref{sec-diag-based-optim}) can be used with $W^{(lr)}$, $W_{\text{in}}^{(lr)}$ and $W_{\text{fb}}^{(lr)}$.

\subsection{Learning $W_\text{out}$}

Given input-output pairs $\{(u(t), y(t))\}_{t=1,\ldots,T}$, where $T$ is the total number of timesteps used for training, we construct:
\begin{align}
    X(t) &= \begin{bmatrix}
        1 & y(t-1) & r(t)
    \end{bmatrix} \in \mathbb{R}^{N'} \\
    Y(t) &= [y(t)] \in \mathbb{R}^{D_{\text{out}}}
\end{align}

\noindent $W_\text{out}$ is computed using standard ridge regression with $\alpha$ the regularization parameter:
\begin{equation}
    W_{out} = (X^T X + \alpha I)^{-1} X^T Y
\end{equation}

\subsection{Computational Complexity}

Here, we analyze each major component of the ESN architecture in terms of time complexity, explicitly accounting for sparsity ($1-\text{connectivity}$) where relevant.

\begin{itemize}
    \item \textbf{$W_{\text{in}}$ and $W_{\text{fb}}$ generation:} One of the input weight matrix $W_{\text{in}} \in \mathbb{R}^{D_{\text{in}} \times N}$ or the feedback weight matrix is optional. These matrices are typically generated with a connectivity $c_{\text{in}}$ or $c_{\text{fb}}$ (where each weight represents a non-zero connection with probability $c_{\text{in}}$ or $c_{\text{fb}}$). The non-zero weights are sampled from a specific distribution (e.g., normal or uniform). The complexity of generating these matrices is $\mathcal{O}(D_{\text{in}} N)$ and $\mathcal{O}(D_{\text{out}} N)$.

    \item \textbf{$W$ generation and Spectral Radius scaling:} The reservoir weight matrix $W \in \mathbb{R}^{N \times N}$ is similarly generated with a connectivity $c_{r}$. A critical step is scaling $W$ to achieve the desired spectral radius $\rho_0$. The computational cost depends heavily on the algorithm used to compute the leading eigenvalues: For dense matrices, standard eigenvalue decomposition algorithms are $\mathcal{O}(N^3)$. For sparse matrices, iterative methods such as the Implicitly Restarted Arnoldi Method (IRAM) are commonly use \citep{lehoucq1998arpack}. The complexity depends on the sparsity and convergence rate, scaling between $\mathcal{O}(N^2)$ and $\mathcal{O}(N^3)$.

    \item \textbf{Reservoir step:} This recursive update involves matrix-vector multiplications. By exploiting sparse matrix operations, the complexity per time step is reduced to the number of non-zero elements: $\mathcal{O}(c_r N^2 + c_{\text{in}} D_{\text{in}} N + c_{\text{fb}} D_{\text{out}}N)$ per time step.
    Therefore, the total cost over $T$ time steps is $\mathcal{O}(T (c_r N^2 + c_{\text{in}} D_{\text{in}} N + c_{\text{fb}} D_{\text{out}}N))$.

    \item \textbf{Readout Layer:} The output computation $y(t) = X(t) W_{\text{out}}$ involves a dense matrix multiplication. With $W_{\text{out}} \in \mathbb{R}^{N' \times D_{\text{out}}}$  , the complexity is $\mathcal{O}(N' D_{{\text{out}}})$. This cost is negligible compared to reservoir dynamics in long sequences, since $D_{\text{out}} \ll N$.

\end{itemize}

\noindent The computational complexity of ESN primarily depends on the reservoir step. This cost dominates due to the recurrent nature of the update, which must be computed sequentially $T$ times. However, it is important to note that the use of sparse matrices can dramatically reduce the effective time complexity compared to dense ones. 


\section{Diagonalization-Based Optimization}
\label{sec-diag-based-optim}

In this section, we introduce an optimization leveraging the diagonalization of the reservoir matrix $W$. By decomposing $ W = PDP^{-1} $, we reformulate the reservoir step into a pointwise form, where each state component evolves independently according to its corresponding eigenvalue. Recall that the set of diagonalizable matrices forms a set of full Lebesgue measure in the space of all square matrices; therefore, a randomly chosen matrix is diagonalizable with probability one \citep{horn2012matrix}. This simplification drastically reduces the amount of computation needed to compute the activity and output of a linear ESN.

\subsection{Core Transformation}

\begin{theorem}{Change-of-basis of the ESN dynamics}
\label{theorem_1}

(i) Let $W$ be the reservoir matrix and $P \in \text{GL}_N(\mathbb{C})$ a basis. Let us denote as $[\cdot]_{_P}$ the transformation into the basis $P$, and express both weights and state of the ESN into this new basis. 

\begin{center}
\begin{tabular}{lll}
 \text{Transformed reservoir matrix} & $[W]_{_P} = P^{-1}\, W P$   \\
 Transformed reservoir state & $[r(t)]_{_P} = r(t) P$   \\
 Transformed input weights & $[W_{\text{in}}]_{_P} = W_{\text{in}}\, P$   \\
 Transformed feedback weights & $[W_{\text{fb}}]_{_P} = W_{\text{fb}}\, P$   \\\\
 Transformed output weights & \multicolumn{2}{l}{
    $[W_\text{out}]_{_P} = 
    \begin{bmatrix}
        W_\text{out,bias} \\
        W_\text{out,out} \\
        [W_\text{out,res}]_{_P} \stackrel{\text{def}}{=} P^{-1}\, W_\text{out,res}
    \end{bmatrix} = 
    \begin{bmatrix}
        I & 0 \\
        0 & P^{-1}
    \end{bmatrix} 
    W_\text{out}$
}
\end{tabular}
\end{center}

(ii) Then the reservoir step becomes:
\begin{align}
    [r(0)]_{_P} &= 0 \\
    [r(t)]_{_P} &= [r(t-1)]_{_P}\, [W]_{_P} + u(t)\, [W_\text{in}]_{_P} + y(t-1)\, [W_\text{fb}]_{_P}
\end{align}

(iii) And the readout step becomes:
\begin{align}
    y(t) &= W_\text{out,bias} + y(t-1)\, W_\text{out,out} + [r(t)]_{_P}\, [W_\text{out,res}]_{_P}
\end{align}

(iv) Using the notation of
\( X(t) =
\begin{bmatrix}
  1  & y(t-1) & r(t)
\end{bmatrix}\),
we define:

\begin{align}
    [X(t)]_{_P} = \begin{bmatrix}
      1 & y(t-1) & [r(t)]_{_P} 
    \end{bmatrix} = X \begin{bmatrix}
      I & 0 \\
      0 & P
    \end{bmatrix}
\end{align}

It is then possible to learn the readout weights directly in $P$: 

\begin{align}
    [W_\text{out}]_{_P} = \left( {[X]_{_P}}^T\, [X]_{_P} + \alpha \begin{bmatrix}
        I & 0 \\
        0 & P^T P
    \end{bmatrix} \right)^{-1} {[X]_{_P}}^T\, Y
\end{align}
\end{theorem}

\begin{proof}
\begin{align*}
    [r(t)]_{_P}
    &= r(t)\, P \\
    &= r(t-1)\, W\, P + u(t)\, W_\text{in}\, P + y(t-1)\, W_\text{fb}\, P \\
    &= r(t-1)\, P\, [W]_{_P}\, P^{-1} P + u(t) [W_\text{in}]_{_P} + y(t-1)\, [W_\text{fb}]_{_P} \\
    &= [r(t-1)]_{_P}\, [W]_{_P} + u(t)\, [W_\text{in}]_{_P} + y(t-1)\, [W_\text{fb}]_{_P}
\end{align*}

\begin{align*}
    y(t)
    &= W_\text{out,bias} + y(t-1)\, W_\text{out,out} + r(t)\, W_\text{out,res} \\
    &= W_\text{out,bias} + y(t-1)\, W_\text{out,out} + r(t)\, P\, P^{-1}\, W_\text{out,res} \\
    &= W_\text{out,bias} + y(t-1)\, W_\text{out,out} + [r(t)]_{_P}\, [W_\text{out,res}]_{_P}
\end{align*}

\begin{align*}
    [W_\text{out}]_{_P} &=
    \begin{bmatrix} I & 0 \\ 0 & P^{-1} \end{bmatrix}\,
    W_\text{out} \\
    &= \begin{bmatrix} I & 0 \\ 0 & P^{-1} \end{bmatrix}\, (X^T X + \alpha I)^{-1}\, X^T\, Y \\
    &= \begin{bmatrix} I & 0 \\ 0 & P^{-1} \end{bmatrix}\, \left(\left([X]_{_P}\, \begin{bmatrix} I & 0 \\ 0 & P^{-1} \end{bmatrix}\right)^T\, \left([X]_{_P}\, \begin{bmatrix} I & 0 \\ 0 & P^{-1} \end{bmatrix}\right) + \alpha I\right)^{-1} \left([X]_{_P}\, \begin{bmatrix} I & 0 \\ 0 & P^{-1} \end{bmatrix}\right)^T Y \\
    &= \begin{bmatrix} I & 0 \\ 0 & P^{-1} \end{bmatrix}\, \left(\begin{bmatrix} I & 0 \\ 0 & (P^T)^{-1} \end{bmatrix} {[X]_{_P}}^T\, [X]_{_P} \begin{bmatrix} I & 0 \\ 0 & P^{-1} \end{bmatrix} + \alpha I\right)^{-1}\, \begin{bmatrix} I & 0 \\ 0 & (P^T)^{-1} \end{bmatrix} {[X]_{_P}}^T\, Y \\
    &= \begin{bmatrix} I & 0 \\ 0 & P \end{bmatrix}^{-1}\, \left(\begin{bmatrix} I & 0 \\ 0 & P^T \end{bmatrix}^{-1} {[X]_{_P}}^T\, [X]_{_P} \begin{bmatrix} I & 0 \\ 0 & P \end{bmatrix}^{-1} + \alpha I\right)^{-1}\, \begin{bmatrix} I & 0 \\ 0 & P^T \end{bmatrix}^{-1} {[X]_{_P}}^T\, Y \\
    &= \left(\begin{bmatrix} I & 0 \\ 0 & P^T \end{bmatrix}\, \left(\begin{bmatrix} I & 0 \\ 0 & P^T \end{bmatrix}^{-1} {[X]_{_P}}^T\, [X]_{_P} \begin{bmatrix} I & 0 \\ 0 & P \end{bmatrix}^{-1} + \alpha I  \right)\begin{bmatrix} I & 0 \\ 0 & P \end{bmatrix}\right)^{-1}\, {[X]_{_P}}^T\, Y \\
    &= \left(\begin{bmatrix} I & 0 \\ 0 & P^T \end{bmatrix}\, \begin{bmatrix} I & 0 \\ 0 & P^T \end{bmatrix}^{-1} {[X]_{_P}}^T\, [X]_{_P} \begin{bmatrix} I & 0 \\ 0 & P \end{bmatrix}^{-1}\begin{bmatrix} I & 0 \\ 0 & P \end{bmatrix} + \alpha \begin{bmatrix} I & 0 \\ 0 & P^T \end{bmatrix}\, \begin{bmatrix} I & 0 \\ 0 & P \end{bmatrix}\right)^{-1}\, {[X]_{_P}}^T\, Y \\
    &= \left({[X]_{_P}}^T\, [X]_{_P} + \alpha \begin{bmatrix} I & 0 \\ 0 & P^T\, P \end{bmatrix}\right)^{-1} {[X]_{_P}}^T\, Y
\end{align*}
\end{proof}

\subsection{A special case of a diagonalizable matrix}

\begin{corollary}{Diagonalization Optimization}
\label{theorem_2}

\noindent Assume the reservoir matrix $W$ is diagonalizable. Let $W = P\, [W]_{_P}\, P^{-1}$ where:
\begin{itemize}
    \item $P \in \text{GL}_N(\mathbb{C})$: Matrix of eigenvectors
    \item $[W]_{_P} = \text{diag}(\Lambda \stackrel{\text{def}}{=} (\lambda_1, \ldots, \lambda_N))$: Diagonal matrix of eigenvalues
\end{itemize}

\noindent Then the reservoir step become pointwise:
\begin{align}
    [r(t)]_{_P} &= [r(t-1)]_{_P} \odot \Lambda + u(t)\, [W_\text{in}]_{_P} + y(t-1)\, [W_\text{fb}]_{_P}
\end{align}
where $\odot$ denotes element-wise multiplication.
\end{corollary}

\begin{proof}
Applying the change of basis to the standard reservoir update equation yields $$[r(t)]_{_P} = [r(t-1)]_{_P} [W]_{_P} + u(t)\, [W_\text{in}]_{_P} + y(t-1)\, [W_\text{fb}]_{_P}$$ Since $[r(t-1)]_{_P}$ is a row vector and $[W]_{_P} = \text{diag}(\Lambda)$, their matrix product strictly simplifies to the element-wise multiplication $[r(t-1)]_{_P} \odot \Lambda$.
\end{proof}

\noindent \textbf{Note on parallelization: } This form enables parallelization since the element-wise operations are independent across neurons.

\subsection{Apply $W_{\text{in}}$ after Reservoir Steps}

This subsection demonstrates that with our method the input matrix $W_{\text{in}}$ can actually be applied after the temporal update. This proves that the network depends merely on $\Lambda$, further justifying our focus on optimizing the spectral properties. For what follows, we suppose a diagonal reservoir without feedback weights ($W_{\text{fb}} = 0$) and with an initial state $r(0) = 0$.

\begin{lemma}
    \begin{equation*}
        r(t) = \sum_{i=1}^t (u(i) W_\text{in}) \odot \Lambda^{t-i}
    \end{equation*}
    where the power of $\Lambda$ denotes element-wise exponentiation, not a matrix product.
\end{lemma}

\begin{proof}
    By recurrence on $t$. 
    For $t=0$: $r(0) = 0$ by definition. The property holds.
    Assume the property is true for a given $t \geq 0$. For $t+1$:
    \begin{align*}
        r(t+1) &= r(t) \odot \Lambda + u(t+1)W_{\text{in}} \\
        &= \left( \sum_{i=1}^t (u(i) W_\text{in}) \odot \Lambda^{t-i} \right) \odot \Lambda + u(t+1)W_{\text{in}} \odot \Lambda^0 \\
        &= \sum_{i=1}^t (u(i) W_\text{in}) \odot \Lambda^{t+1-i} + u(t+1)W_{\text{in}} \odot \Lambda^{t+1-(t+1)} \\
        &= \sum_{i=1}^{t+1} (u(i) W_\text{in}) \odot \Lambda^{t+1-i}
    \end{align*}
    The property holds for $t+1$, concluding the proof.
\end{proof}

\noindent According to this formula, to compute the reservoir state, we first apply $W_\text{in}$ to the input signal and then apply $\Lambda$. Actually, this can be done in the reverse order. Let $R(t) \in \mathbb{C}^{D_\text{in} \times N}$ be a state matrix defined by:
\begin{align*}
    R(0) &= 0 \\
    R(t) &= R(t-1) \odot 
    \begin{bmatrix} 
        \Lambda \\ 
        \vdots \\ 
        \Lambda 
    \end{bmatrix} 
    + 
    \begin{bmatrix} 
        u(t)^T & \dots & u(t)^T
    \end{bmatrix}
\end{align*}
Here, the first term broadcasts the row vector $\Lambda$ across all $D_\text{in}$ rows, and the second term repeats the column vector $u(t)^T$ across all $N$ columns to form a $D_\text{in} \times N$ matrix.

\begin{lemma}
    Let $\text{row}_d(\cdot)$ denote the $d$-th row of a matrix.
   \begin{equation*}
        \text{row}_d(R(t)) = \sum_{i=1}^t u(i)_d \Lambda^{t-i}
    \end{equation*}
\end{lemma}

\begin{proof}
    By recurrence on $t$ for a fixed row $d$.
    For $t=0$: $\text{row}_d(R(0)) = 0$ by definition. The property holds.
    Assume true for $t \geq 0$. For $t+1$:
    \begin{align*}
        \text{row}_d(R(t+1)) &= \text{row}_d(R(t)) \odot \Lambda + u(t+1)_d \Lambda^0 \\
        &= \left( \sum_{i=1}^t u(i)_d \Lambda^{t-i} \right) \odot \Lambda + u(t+1)_d \Lambda^0 \\
        &= \sum_{i=1}^{t+1} u(i)_d \Lambda^{t+1-i}
    \end{align*}
\end{proof}

\begin{theorem}
    \label{th:win_after}
    The standard reservoir state $r(t)$ can be recovered by element-wise weighting and summation:
    \begin{equation*}
        r(t) = \sum_{d=1}^{D_\text{in}} \text{row}_d(W_\text{in}) \odot \text{row}_d(R(t))
    \end{equation*}
    Or more compactly, letting $\mathbf{1}$ be a column vector of ones of size $D_\text{in} \times 1$:
    \begin{equation*}
        r(t) = \mathbf{1}^T (W_\text{in} \odot R(t))
    \end{equation*}
\end{theorem}

\begin{proof}
Using the linearity of the summation and detailing the intermediate associative steps:
\begin{align*}
    r(t) &= \sum_{i=1}^t (u(i)W_\text{in}) \odot \Lambda^{t-i} \\
    &= \sum_{i=1}^t \left( \sum_{d=1}^{D_\text{in}} u(i)_d \text{row}_d(W_\text{in}) \right) \odot \Lambda^{t-i} \\
    &= \sum_{i=1}^t \sum_{d=1}^{D_\text{in}} \left( u(i)_d \text{row}_d(W_\text{in}) \right) \odot \Lambda^{t-i} \\
    &= \sum_{d=1}^{D_\text{in}} \sum_{i=1}^t \text{row}_d(W_\text{in}) \odot \left( u(i)_d \Lambda^{t-i} \right) \\
    &= \sum_{d=1}^{D_\text{in}} \text{row}_d(W_\text{in}) \odot \left( \sum_{i=1}^t u(i)_d \Lambda^{t-i} \right) \\ 
    &= \sum_{d=1}^{D_\text{in}} \text{row}_d(W_\text{in}) \odot \text{row}_d(R(t)) \\
    &= \mathbf{1}^T (W_\text{in} \odot R(t))
\end{align*}
\end{proof}

\subsubsection{Application: Diagonal Linear ESN}
\label{sec:win_end}

Let $(\Lambda \in \mathbb{C}^N, W_\text{in} \in \mathbb{C}^{D_\text{in} \times N}, W_\text{out} \in \mathbb{C}^{N \times D_\text{out}})$ be a diagonal linear ESN, fed with an input signal $(u(t))_{1\leq t \leq T}$. According to the previous theorem:
\begin{equation*}
    y(t) = r(t) W_\text{out} = \left[ \mathbf{1}^T (W_\text{in} \odot R(t)) \right] W_\text{out}
\end{equation*}
While this formulation is mathematically exact for any dimensions, optimizing $W_\text{out}$ directly from $R(t)$ without fixing $W_\text{in}$ first remains complex when $D_\text{in} > 1$ or $D_\text{out} > 1$. However, this establishes a fundamental property: the general temporal dynamics are entirely captured by $R(t)$ and $\Lambda$, completely independently of the input weight $W_\text{in}$. This property becomes highly actionable when $D_\text{in} = 1$ and $D_\text{out} = 1$ (see Appendix \ref{app:special_case_win}).


\subsection{Computational Advantages}

This approach reduces computational costs after initial preprocessing \citep{golub2013matrix}:
\begin{itemize}
    \item \textbf{Eigendecomposition}: $\mathcal{O}(N^3)$ (one-time preprocessing cost), compared to the spectral radius computation, which can be between $\mathcal{O}(N^2)$ and $\mathcal{O}(N^3)$ as explained above.
    \item \textbf{$[W_\text{in}]_{_P}$ and $[W_\text{fb}]_{_P}$ computation}: $\mathcal{O}(N^2D_{\text{in}})$ and $\mathcal{O}(N^2D_{\text{fb}})$
    \item \textbf{Transformed reservoir step}: $\mathcal{O}(N(D_{\text{in}}+D_{\text{out}}))$ (reduced from $\mathcal{O}(N(c_\text{r}N + c_\text{in} D_{\text{in}} + c_\text{fb} D_{\text{out}})$) with $D_{\text{in}},D_{\text{out}} \ll N$ (reservoir looking to project the inputs into a high dimensional space).
    \item \textbf{Transformed readout step}: $\mathcal{O}(N' D_{\text{out}})$ (unchanged)
\end{itemize}

\noindent The key advantage is the reduction of the reservoir step complexity from $\mathcal{O}(N^2)$ to $\mathcal{O}(N(D_{\text{in}}+D_{\text{out}})$ per time step, making this approach highly efficient despite the initial $\mathcal{O}(N^3)$ preprocessing cost. Figure \ref{fig:step_duration} shows the comparison between standard linear reservoir computation and our optimisation in function of the number of neurons in the reservoir for each processing step. We can also note that even with the use of complex numbers (which should doubles the number of parameters with a real and an imaginary part), the training of the readout can be performed with real matrices (see Appendix \ref{app:implementation}) as in the standard computations, leading to no additional cost.

\begin{figure}[H]
    \centering
    \includegraphics[width=1\linewidth]{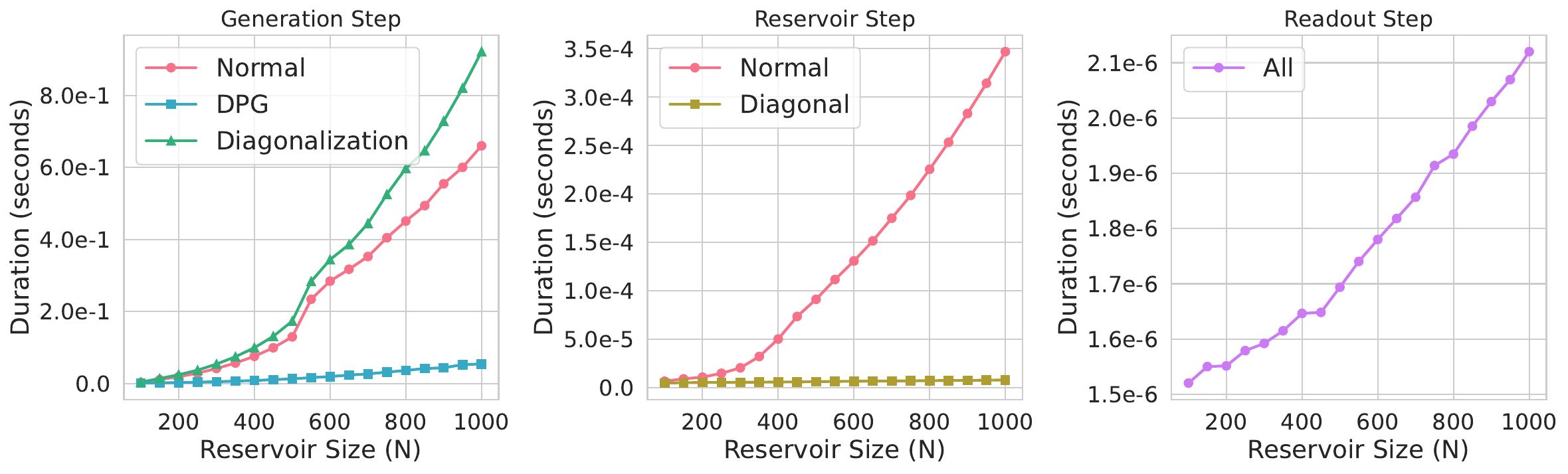}
    \caption{Comparison of standard computation and proposed optimizations for different reservoir sizes. 
    (i) Generation Step: Evaluates three initialization methods. \textit{Normal} generates a standard linear reservoir with a weight matrix $W$. \textit{Diagonalization} generates a standard $W$ and then diagonalizes it (applicable to EWT/EET in section \ref{sec:methods}). \textit{DPG} directly generates a diagonal matrix (see section \ref{sec:methods}).
    (ii) Reservoir Step: Compare the standard computation with the \textit{Diagonal} one. Because EWT, EET, and DPG share an identical diagonal structure after the generation step, they are represented together as \textit{Diagonal}.
    (iii) Readout Step: Displays only a single curve because the computational cost is identical across all methods. With the implementation proposed in Appendix \ref{app:implementation}, the training of the readout can be performed with real matrices, equating the readout cost of the standard method.
    (Note) For reservoir and readout steps, the duration displayed is for a single time step. To evaluate the total duration over the whole sequence, these values must be multiplied by the total number of time steps.}
    \label{fig:step_duration}
\end{figure}

\section{Methods}
\label{sec:methods}

This section presents three methods for efficient linear reservoir design and training. The first, Eigenbasis Weight Transformation (EWT), enables efficient inference by transforming pre-trained weights $W_{\text{out}}$ into the eigenbasis without retraining. The second, End-to-End Eigenbasis Training (EET), enables full end-to-end training by directly computing $[W_\text{out}]_{_P}$ in the transformed space, which drastically reduces computational cost during training. Finally, Direct Parameter Generation (DPG) eliminates the need for initial reservoir matrix $W$ by directly constructing its spectral components $\Lambda$ and eigenvector basis $P$ and then computing $[W_\text{out}]_{_P}$ in the transformed space.

\subsection{Training Data Preparation}

Given input-output pairs $\{(u(t), y(t))\}_{t=1,\ldots,T}$, we construct:
\begin{align}
    X(t) &= \begin{bmatrix}
        1 & y(t-1) & r(t)
    \end{bmatrix} \in \mathbb{R}^{N'} \quad \text{(extended state)} \\
    Y(t) &= [y(t)] \in \mathbb{R}^{D_{\text{out}}} \quad \text{(target output)}
\end{align}
where $T$ is the length of the sequence.

\subsection{Method 1: Eigenbasis Weight Transformation (EWT)}

If $W_{out}$ has already been computed using standard ridge regression:
\begin{equation}
    W_{out} = (X^T X + \alpha I)^{-1} X^T Y
\end{equation}

\noindent Then the transformed output weights are (see Theorem 1 in \ref{theorem_2}):
\begin{equation}
    [W_\text{out}]_{_P} = \begin{bmatrix}
        I & 0 \\
        0 & P^{-1}
    \end{bmatrix}\, W_\text{out}
\end{equation}

\subsection{Method 2: End-to-End Eigenbasis Training (EET)}

If $W_\text{out}$ has not already been computed, it is possible to train directly $[W_\text{out}]_{_P}$.

\vspace{0.5cm}
\noindent \textbf{Proposition: Direct Output Weight Computation} \\
\vspace{-0.4cm}

\noindent Let
\([X]_{_P} = \begin{bmatrix}
    1 & y(t-1) & [r(t)]_{_P}
\end{bmatrix} \in\mathbb{C}^{N'}\) be the concatenated transformed reservoir states. It is then possible to obtain $W_\text{out}$ by training $[W_\text{out}]_{_P}$ directly.

\begin{align}
    [W_{\text{out}}]_{_P} = \left( {[X]_{_P}}^T\, [X]_{_P} + \alpha \begin{bmatrix}
        I & 0 \\
        0 & P^T P
    \end{bmatrix} \right)^{-1} {[X]_{_P}}^T\, Y
\end{align}

\subsection{Method 3: Direct Parameter Generation (DPG)}

Rather than generating $(W, W_\text{in})$ and then diagonalizing like the EET method, this section explores directly generating $(\Lambda, [W_\text{in}]_{_P}, P)$ to avoid the initial diagonalization cost. We called that method Direct Parameter Generation (DPG). Notably, our results with DPG aligns with the findings of \citet{rodan2010minimum} that minimal and deterministic topologies can replace dense random matrices. 

\begin{figure}[H]
    \centering
    \includegraphics[width=1.0\textwidth]{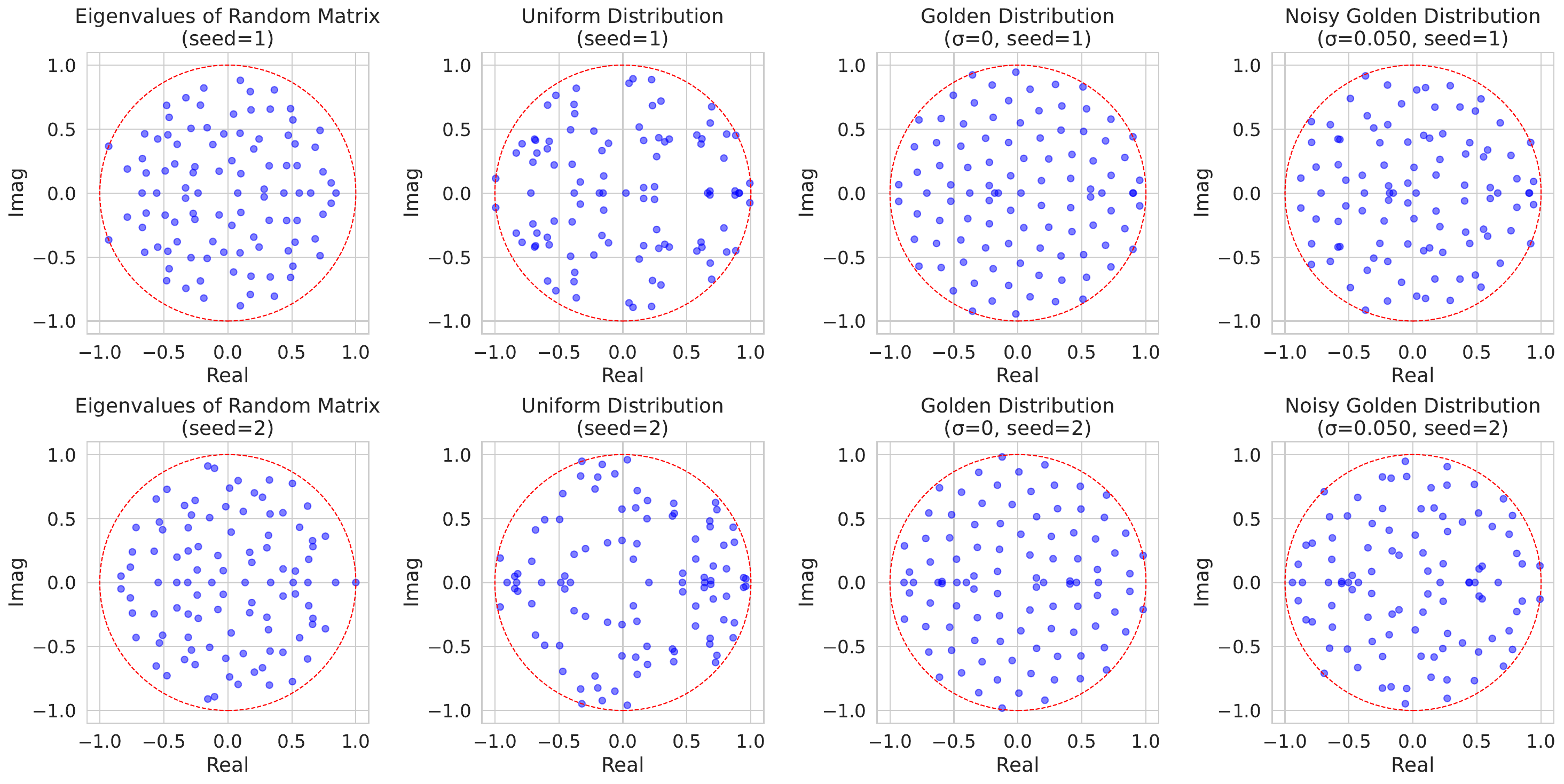}
    \caption{Comparison of eigenvalue distributions in the complex plane. On the first column (left) the spectrum derived from a standard random reservoir matrix $W$. On the second column the spectrum generated via \textit{Uniform Distribution} (Algorithm \ref{alg:gen_eigenvalues}). On the third column the spectrum generated via the deterministic \textit{Golden Distribution} method (Algorithm \ref{alg:golden_eigenvalues}) without noise, and on the fourth column the specturm generated via the \textit{Noisy Golden Distribution}. We observe that the \textit{Noisy Golden distribution} (right) achieves a significantly more homogeneous coverage of the unit disk than the \textit{Uniform Distribution} or its non-noisy counterpart, effectively matching the spectral density of the standard reservoir (left).}
    \label{fig:eigenvalues_distributions}
\end{figure}

\subsubsection{Uniform Distribution}

\noindent The spectral structure of the reservoir weight matrix $W$ plays a fundamental role in determining both the stability and the dynamics of its activity. A well-designed spectrum enables rich temporal dynamics, essential for tasks involving sequence modeling and prediction. Drawing upon insights from random matrix theory (particularly the classical work of \cite{edelman1995many}), we observe that for an $N \times N$ random matrix with i.i.d.\ Gaussian entries, the expected number of real eigenvalues scales as:
\begin{align}
\mathbb{E}[N_{\mathrm{real}}] \sim \sqrt{\frac{2N}{\pi}} 
\end{align}

\noindent We propose an approach to directly construct $P$ and $\Lambda$ without the need of generating $W$ by explicitly engineering the eigenvector matrix $P$ such that its associated eigenvalue spectrum $\Lambda$ matches a desired distribution with a controlled number of real and complex conjugate pairs. To this end, we introduce two complementary algorithms that respectively generate a valid $\Lambda$ and $P$.

\begin{algorithm}
\caption{Random generation of the eigenvalues}
\label{alg:gen_eigenvalues} 
\begin{algorithmic}[1]
\Function{RandomEigenvalues}{$N$, $\textrm{sr}$}
    \State $N_{real} \gets \left\lfloor \sqrt{\frac{2N}{\pi}} \right\rfloor$
    \If{$N_{real} \neq n \mod 2$}
        \State $N_{real} \gets N_{real} + 1$
    \EndIf
    \State $\Lambda_{real}$: array of size $N_{real}$
    \State $R2$: array of size $(N-N_{real})/2$
    \State $\Theta$: array of size $(N-N_{real})/2$
    \State $\Lambda_{cpx}$: array of size $(N-N_{real})/2$
    \State Sample $\Lambda_{real} \sim \text{Uniform}(-\textrm{sr}, \textrm{sr})$
    \State Sample $U \sim \text{Uniform}(0,1)$
    \State Sample $\Theta \sim \text{Uniform}(0, \pi)$
    \State $\Lambda_{cpx} \gets \textrm{sr} \cdot \sqrt{U} \cdot \exp(i\cdot\Theta)$
    \State \Return $\Lambda_{real}, \Lambda_{cpx}, \text{conj}(\Lambda_{cpx})$
\EndFunction
\end{algorithmic}
\end{algorithm}

Algorithm \ref{alg:gen_eigenvalues} generates a structured, random eigenvalue spectrum $\Lambda = (\lambda_1,\dots,\lambda_N)$ by balancing the number of real and complex conjugate pairs. 
Specifically, it ensures that exactly $N_{\text{real}}$ eigenvalues are real and since $W$ is real, the remaining ($N-N_{\text{real}}$) eigenvalues are $(N - N_{\text{real}})/2$ complex conjugate pairs.  
To guarantee this property, the algorithm enforces that $N_{\text{real}}$ and $N$ have the same parity (by incrementing $N_{\text{real}}$ by one if necessary), so that $N - N_{\text{real}}$ is always even and the complex eigenvalues can consistently be arranged into conjugate pairs. 
The radial component of each complex pair is sampled from $\sim \sqrt{\mathcal{U}}$ (square root of a uniform distribution), while the angular component is uniformly drawn from $[0, \pi)$. An example of this distribution is provided in Figure \ref{fig:eigenvalues_distributions}.

\begin{algorithm}
\caption{Random generation of the eigenvectors}
\label{alg:gen_eigenvectors} 
\begin{algorithmic}[1]
\Function{RandomEigenvectors}{$N, N_{real}$}
    \State $P$: array of shape $(N, N)$
    \For{$i \gets 0$ to $N_{real} - 1$}
        \State Sample $v \sim \text{Normal}(0,1)$
        \State $v \gets v/\lVert v \rVert_2$
        \State $P[:, i] \gets v$
    \EndFor
    \For{$k \gets 0$ to $(N-N_{real})/2 - 1$}
        \State Sample $v_R \sim \text{Normal}(0,1)$
        \State Sample $v_I \sim \text{Normal}(0,1)$
        \State $v \gets (v_R + iv_I)/\lVert v_R + iv_I \rVert_2$
        \State $P[:, N_{real}+k] \gets v$
        \State $P[:, (N+N_{real})/2+k] \gets \text{conj}(v)$
    \EndFor
    \State \Return P
\EndFunction
\end{algorithmic}
\end{algorithm}

Algorithm \ref{alg:gen_eigenvectors} constructs a random matrix of eigenvectors $P \in \mathrm{GL}_n(\mathbb{C})$. Since $W$ is real, we sample eigenvectors associated with real eigenvalues and with complex-conjugate eigenvalue pairs in different ways. For real eigenvalues, we sample real-valued eigenvectors with entries drawn from a standard Gaussian distribution. For complex-conjugate eigenvalue pairs, we sample the first eigenvector as a complex-valued vector whose real and imaginary parts have entries drawn from a standard Gaussian distribution, and the second one as its conjugate. This sampling ensures that the resulting set of vectors forms a basis so that $P$ is invertible.

\subsubsection{Golden Distribution}

\begin{algorithm}
\caption{Random generation of the golden eigenvalues}
\label{alg:golden_eigenvalues} 
\begin{algorithmic}[1]
\Function{GoldenEigenvalues}{$N$, $\textrm{sr}$}
    \State $N_\text{real} \gets \left\lfloor \sqrt{\frac{2N}{\pi}} \right\rfloor$
    \State $N_\text{real} \gets (N - N_\text{real}) \mod 2$
    \State $N_\text{cpx} \gets \frac{N - N_\text{real}}{2}$
    \\
    \State $\Lambda_{real}$: array of size $N_{real}$
    \State Sample $\Lambda_{real} \sim \text{Uniform}(-1,1)$
    \\
    \State $\Lambda_\text{cpx}$: array of size $N_\text{cpx}$
    \State $j \gets 0$
    \State $k \gets 0$
    \State $v \gets \text{Uniform}(0,2)$
    \While{$j < N_\text{cpx}$}
        \State $k \gets k + 1$
        \State $v \gets (v + 3 - \sqrt{5}) \mod 2$
        \If{$v < 1$}
            \State $\Lambda_\text{cpx}(j) \gets \sqrt{\frac{k}{2 \, n_\text{cpx}}}\exp(i\pi v)$
            \State $j \gets j + 1$
        \EndIf
    \EndWhile
    \\
    \State $\lambda \gets \frac{\textrm{sr}}{\max(|\Lambda_\text{real}|,|\Lambda_\text{cpx}|)}$
    \State $\Lambda_\text{real} \gets \lambda \cdot \Lambda_\text{real}$
    \State $\Lambda_\text{cpx} \gets \lambda \cdot \Lambda_\text{cpx}$
    \\
    \State $\text{Noise}$: array of size $N_\text{cpx}$
    \State Sample $\text{Noise} \sim \text{Normal}(0,\sigma) + i \cdot \text{Normal}(0,\sigma)$
    \State $\Lambda_\text{cpx} \gets \Lambda_\text{cpx} + \text{Noise}$
    \\
    \State \Return $\Lambda_{real}, \Lambda_{cpx}, \text{conj}(\Lambda_{cpx})$
\EndFunction
\end{algorithmic}
\end{algorithm}

\noindent Algorithm \ref{alg:golden_eigenvalues} introduces a deterministic variation for generating the eigenvalue spectrum, aiming to improve the coverage of the complex plane compared to \textit{Uniform Distribution}. While the partition between real and complex eigenvalues follows the same statistical scaling as Method 3 (where $N_{\text{real}} \approx \sqrt{2N/\pi}$), the placement of complex eigenvalues utilizes a deterministic sequence based on the Golden Ratio. Specifically, the complex eigenvalues with positive imaginary part are arranged in a phyllotaxis spiral pattern (e.g. the pattern of a sunflower, see \citet{vogel1979better}), and the complex eigenvalues with negative imaginary part are their complex conjugates. Eigenvalue are placed in an iterative way, the angular position is updated step by step using the Golden Angle (via the term $3 - \sqrt{5}$), while the magnitude scales with the square root of the index to ensure a constant density distribution across the unit disk. This approach mitigates the clustering and gaps often observed with pseudo-random generators, providing a more homogeneous spectral structure. An example of this distribution is provided in Figure \ref{fig:eigenvalues_distributions}.


\section{Experiments and Results}

In this section, we evaluate our proposed methods on two standard benchmarks: the Multiple Superimposed Oscillators (MSO) and Memory Capacity (MC). Our experiments demonstrate that our approaches (notably Direct Parameter Generation (DPG) combined with the deterministic Golden distribution) achieve equivalent or superior performance compared to standard baselines, while also exposing the structural limitations of matrix diagonalization at extremely low connectivity levels.

\subsection{Multiple Superimposed Oscillators}

\begin{figure}[h]
    \centering
    \includegraphics[width=1\textwidth]{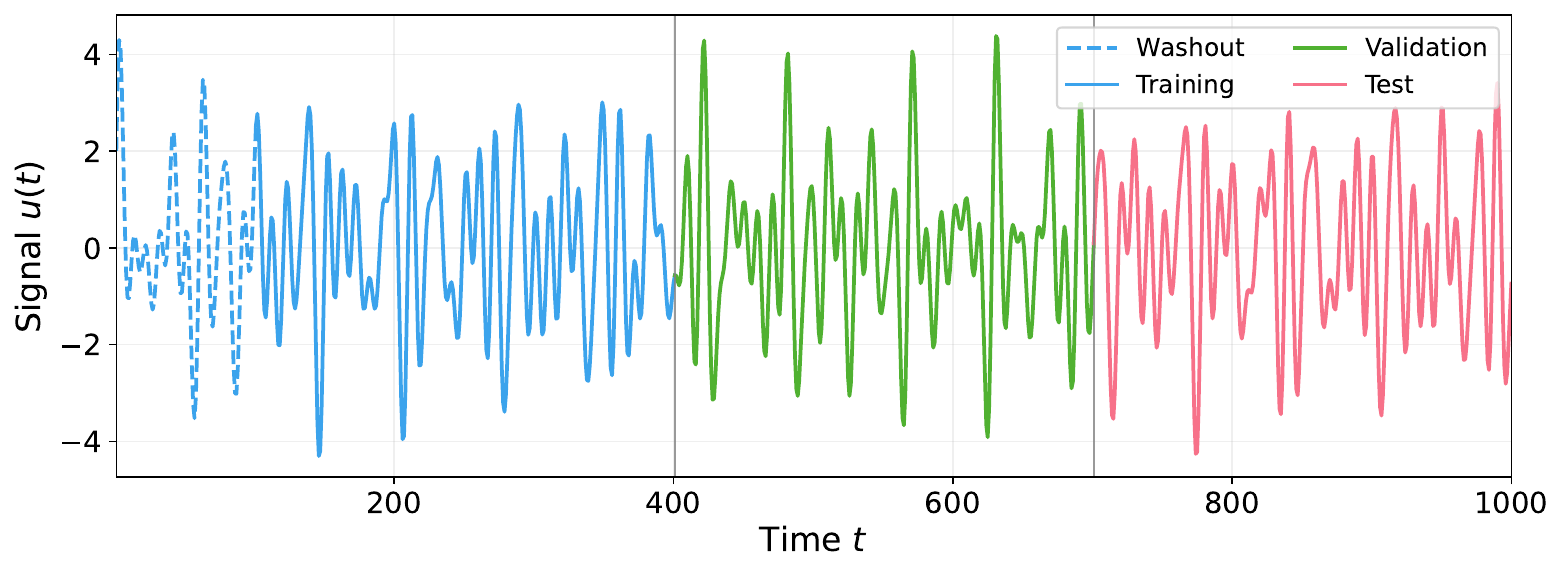}
    \caption{Illustration of the Multiple Superimposed Oscillators (MSO) time series for the $K=5$ task (MSO5). The target signal $U_5(t)$ is generated by summing five distinct  sinusoidal components. The complete sequence of 1000 time steps is partitioned into a training part of 400 steps (blue) which includes an initial 100-step washout used to discard transient reservoir dynamics (dashed blue), a validation of 300 steps (green) for hyperparameter tuning, and a final test of 300 steps (pink) to evaluate predictive performance.}
\end{figure}




The Multiple Superimposed Oscillators (MSO) task \citep{gallicchio2017hierarchical, koryakin2012balanced} is a well-established benchmark used to evaluate the ability of recurrent networks to capture and reproduce complex, multi-frequency temporal dynamics. It consists of a time series generated by summing $K$ discrete-time sinusoidal components:

\begin{equation}
    U_K(t) = \sum_{k=1}^K \sin(\alpha_K\,t)
    \label{eq:mso_regression}
\end{equation}

\noindent where $K$ denotes the total number of superimposed sinusoids and $\alpha_k$ represents their respective angular frequencies. Following the setup introduced by \citet{gallicchio2017hierarchical}, we utilize a defined set of 12 distinct frequencies for the $\alpha_K$ parameters: $\alpha_1 = 0.2, \alpha_2 = 0.331, \alpha_3 = 0.42, \alpha_4 = 0.51, \alpha_5 = 0.63, \alpha_6 = 0.74, \alpha_7 =
0.85, \alpha_8 = 0.97, \alpha_9 = 1.08, \alpha_{10} = 1.19, \alpha_{11} = 1.27, \alpha_{12} = 1.32$.
To systematically evaluate the performance of our models across varying levels of signal complexity, we define 12 distinct sequential tasks based on the number of components $K \in \{1, 2, \dots, 12\}$. For a given task $K$, the target signal $U_K$ is the sum of the first $K$ sinusoids. 
The network is trained to predict the next time step $U_K(t+1)$ given the current ground truth input $U_K(t)$.
For each of the 12 tasks, the generated time series is partitioned into three distinct periods: a training set of $T_{\text{train}} = 400$ time steps, a validation set of $T_{\text{valid}} = 300$ time steps, and a test set of $T_{\text{test}} = 300$ time steps. To mitigate the impact of the arbitrary initial state of the reservoir, the first 100 steps of the training phase are discarded as a washout (or warmup) and are not used to compute the readout weights. To ensure the statistical robustness of our comparative analysis, the evaluation for each task $U_K$ is averaged over 10 seeds. For each test, an exhaustive grid search is conducted to determine the optimal configuration. The complete set of hyperparameter considered is detailed in Table \ref{tab:hyperparams_mso}. Furthermore, by leveraging Theorem \ref{th:win_after}, we significantly accelerate the hyperparameter search for all diagonal reservoir methods. Because the reservoir states can be computed independently of the input scaling factor, we only need to collect these states once per seed rather than recomputing them for each of the three \texttt{input\_scaling} values evaluated in our grid search. This mathematical property effectively divides the state computation time by a factor of three compared to the standard \textit{Normal} baseline.

\begin{table}[H]
\centering
\small
\caption{Hyper-parameters values considered for model selection on the MSO tasks reported from \citet{gallicchio2017hierarchical}}
\begin{tabular}{l|l}
\toprule
\textbf{Hyper-parameter} & \textbf{Values considered} \\
\midrule
Reservoir size ($N$) & 100 \\
Input scaling ($\text{scale}_{\text{in}}$) & 0.01, 0.1, 1 \\
Leaking rate ($\text{lr}$) & 0.1, 0.3, 0.5, 0.7, 0.9, 1.0 \\
Spectral radius ($\rho$) & 0.1, 0.3, 0.5, 0.7, 0.9, 1.0 \\
Ridge regression regularization ($\alpha$) & $10^{-11}, 10^{-10}, \dots, 10^0$ \\
\bottomrule
\end{tabular}
\label{tab:hyperparams_mso}
\end{table}

\begin{table}[H]
\caption{Performance comparison on the Multiple Superimposed Oscillators (MSO) across 10 independent seeds. Metric is the Root Mean Square Error (RMSE). \textit{Normal} corresponds to a standard $W$, \textit{Diagonalized} to the EET method, others corresponds to the DPG method with different initialization strategies.}
\label{tab:results_mso}
\centering
\scriptsize
\begin{tabular}{c|c|c|c|c|c|c}
\toprule
Task & Normal & Diagonalized & Uniform Dist. & Golden Dist. & Noisy Golden Dist.  & Sim Dist. \\
& & & & ($\sigma = 0$) & ($\sigma = 0.2$) & \\
\midrule
MSO1 & 1.65e-14 & \textbf{1.58e-14} & 5.85e-14 & 2.49e-14 & 4.77e-14 & 3.56e-14 \\
MSO2 & 2.55e-13 & 2.78e-13 & 2.28e-13 & \textbf{1.45e-13} & 2.39e-13 & 2.44e-13 \\
MSO3 & 5.42e-12 & 9.14e-12 & \textbf{4.49e-12} & 9.07e-12 & 6.14e-12 & 8.37e-12 \\
MSO4 & 1.39e-10 & 5.77e-10 & 3.64e-10 & 7.22e-11 & \textbf{6.93e-11} & 2.28e-10 \\
MSO5 & 2.75e-09 & 4.03e-08 & 2.95e-08 & \textbf{5.24e-10} & 1.63e-09 & 1.87e-08 \\
MSO6 & \textbf{7.38e-09} & 2.54e-08 & 2.07e-08 & 1.07e-08 & 1.18e-08 & 6.16e-08 \\
MSO7 & \textbf{2.96e-08} & 9.48e-08 & 7.16e-08 & 5.98e-08 & 5.36e-08 & 8.55e-08 \\
MSO8 & \textbf{2.75e-08} & 9.68e-08 & 3.57e-07 & 1.15e-07 & 6.44e-08 & 1.41e-07 \\
MSO9 & \textbf{4.98e-08} & 2.69e-07 & 4.33e-07 & 1.69e-07 & 1.03e-07 & 1.63e-07 \\
MSO10 & 4.65e-07 & 3.32e-07 & 4.15e-07 & 2.31e-07 & \textbf{1.61e-07} & 2.73e-07 \\
MSO11 & 5.62e-07 & 7.38e-07 & 1.85e-06 & 7.49e-07 & \textbf{3.16e-07} & 3.71e-07 \\
MSO12 & 9.71e-07 & 2.98e-06 & 1.34e-06 & 1.01e-06 & \textbf{8.44e-07} & 2.63e-06 \\
\bottomrule
\end{tabular}
\end{table}

The results presented in Table \ref{tab:results_mso} show an interesting distribution of the best scores (in bold) across the different evaluated methods. This diversity suggests that our optimization approaches, particularly those based on Direct Parameter Generation (DPG) with Golden Distribution, perform at least as well as the \textit{Normal} method (standard linear ESN). The slight RMSE variations from one method to another are naturally explained by the structural differences in the weights: since the DPG approach bypasses the generation of an explicit matrix $W$, the eigenvalue spectrum and the eigenvector configuration inevitably differ from those of a standard random matrix.
Across the 12 MSO tasks, \textit{Diagonal Golden Noised ($\sigma = 0.2$)} achieves the best score 4 times (MSO 4, 10, 11, 12), in par with the \textit{Normal} baseline, which also dominates 4 tasks (MSO 6, 7, 8, 9). If we strictly isolate these two methods for a direct comparison, they are in a perfect tie (6 to 6). This strong competitiveness indicates that our strategy of generating a noisy "Golden" spectrum constitutes a highly relevant alternative to classical initialization.
Conversely, it can be noted that the \textit{Diagonal Sim} method is the only one to never attain the top rank, although it closely follows the general trend with scores near those of the leaders. Furthermore, the comparison among DPG methods reveals that the addition of noise to the \textit{Golden Distribution} outperforms the strict \textit{Golden Distribution} version. This phenomenon can be understood intuitively. Visually and statistically, the injection of noise allows the deterministic Golden distribution to smooth its structure and much more faithfully mimic the natural empirical eigenvalue distribution of a true random matrix.
These results validate the DPG approach paired with a \textit{Noisy Golden Distribution}. They empirically confirm that it is possible to directly generate a structured eigenvalue distribution combined with random eigenvectors to guarantee a performance level equivalent, and occasionally superior, to the standard method.

\begin{figure}[H]
    \centering
    \includegraphics[width=1\textwidth]{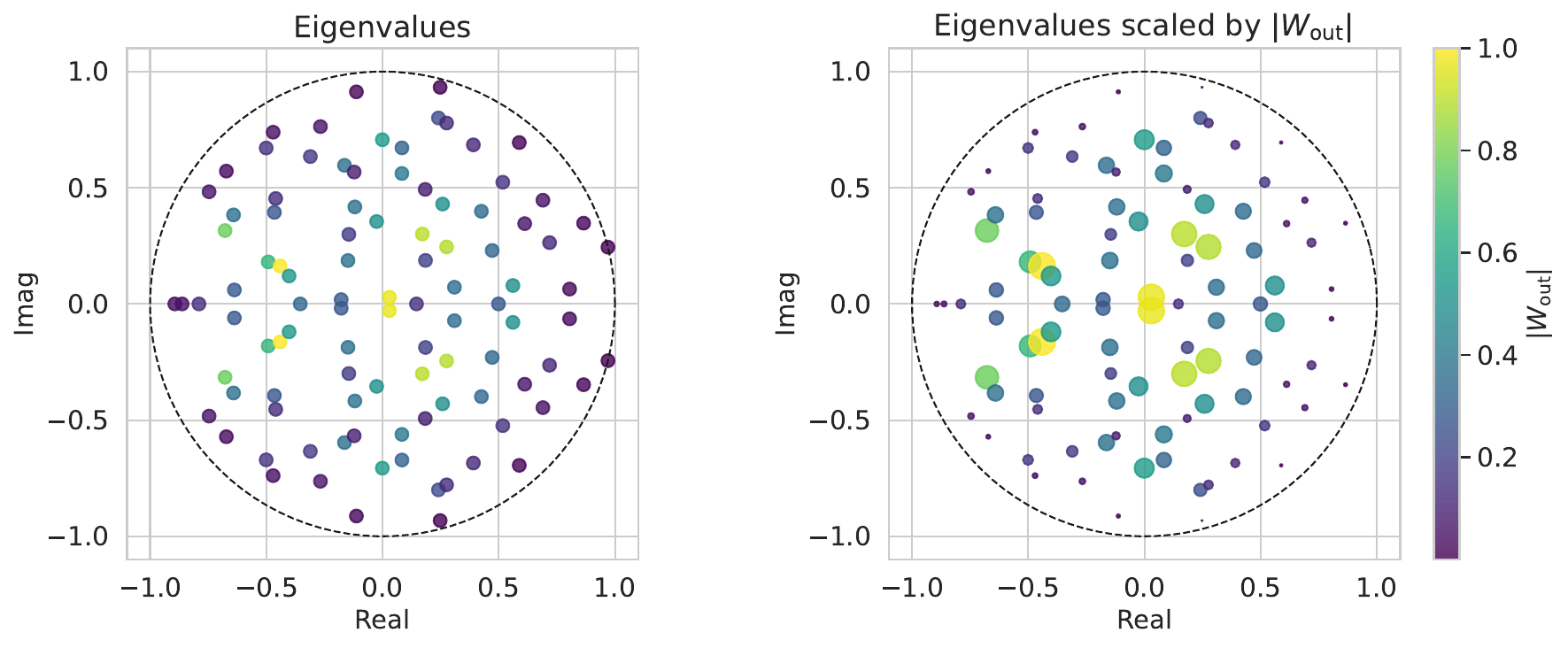}
    \caption{Visualization of spectral importance via readout weights on the MSO task. The reservoir eigenvalues are plotted in the complex plane, where the marker size is proportional to the absolute value of the corresponding weight in $W_{\text{out}}$ (normalized between 0 and 1). Larger points identify the specific eigenvalues that contribute most significantly to the model's prediction for this task.
    Points are the same on the left and on the right, just the size of points is changing: on the right, points close to the circle are non longer visible which may create an illusion that the point distribution was contracted, but that's not the case.}
    \label{fig:eigenvalues_wout}
\end{figure}





\subsection{Memory Capacity}

The Memory Capacity (MC), introduced by Jaeger \cite{jaeger2001short}, is a fundamental metric used to quantify the short-term memory of a recurrent neural network. It measures the reservoir's ability to reconstruct past continuous inputs from its current internal state. For an independent and identically distributed (i.i.d.) input sequence $u(t)$ and a trained output unit $y_k(t)$, the determination coefficient for a specific delay $k$ is given by the squared correlation coefficient:

\begin{equation}
d(u(t-k), y_k(t)) = \frac{\text{cov}^2(u(t-k), y_k(t))}{\sigma^2(u(t))\sigma^2(y_k(t))}
\end{equation}

The $k$-delay short-term memory capacity of the network is then formally defined as the maximum of this determination coefficient optimized over the output weights $w_{k}^{out}$:

\begin{equation}
MC_k = \max_{w_k^{out}} d[w_k^{out}](u(t-k), y_k(t))
\end{equation}

In our experiments, we utilize this established task to test the practical memory capacity of our methods. We use reservoirs without a leaking rate and with the spectral radius strictly set to 1. We then evaluate both the $MC_k$ scores and the computational duration of the different initialization methods across various values of reservoir units $N$ and requested delay $k$.

\begin{figure}[H]
    \centering
    \includegraphics[width=1\textwidth]{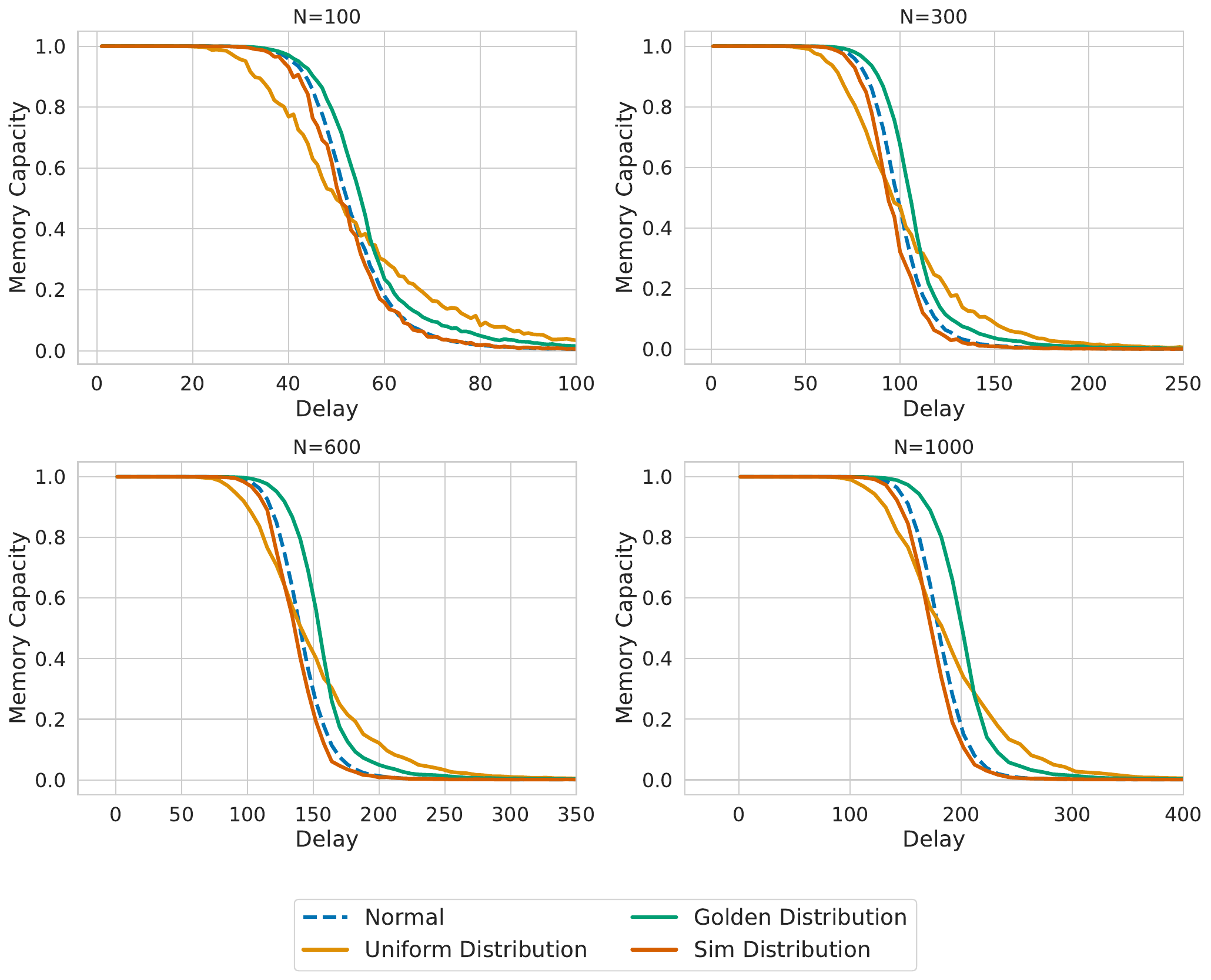}
    \caption{Memory capacity as a function of delay across different reservoir sizes. The four subplots illustrate the memory capacity evaluated against the delay for reservoir sizes N=100, 300, 600, and 1000. As expected, increasing the reservoir size brings an overall improvement in memory capacity. The performance of five different methods is compared: the standard Echo State Network with a classic weight matrix $W$ (\textit{Normal}), Direct Parameter Generation (DPG) utilizing a  \textit{Uniform Distribution}, DPG utilizing the deterministic \textit{Golden Distribution}, and DPG utilizing the \textit{Sim Distibution} which uses eigenvalues extracted from a standard random matrix $W$ combined with randomly generated eigenvectors $P$. Observing the results, the \textit{Diagonal Uniform} method exhibits a more balanced degradation profile than the \textit{Normal} baseline, though they remain roughly equivalent and intersect near a memory capacity of 0.5. Notably, the \textit{Golden Distribution} initialization systematically outperforms the standard \textit{Normal} method across all reservoir sizes. Conversely, the \textit{Sim Distribution} curve closely follows the \textit{Normal} baseline, suggesting that the eigenvectors play a secondary role compared to eigenvalues, even if it demonstrates a small but consistent underperformance.}
    \label{fig:mc_delay}
\end{figure}

\begin{figure}[H]
    \centering
    \includegraphics[width=1\linewidth]{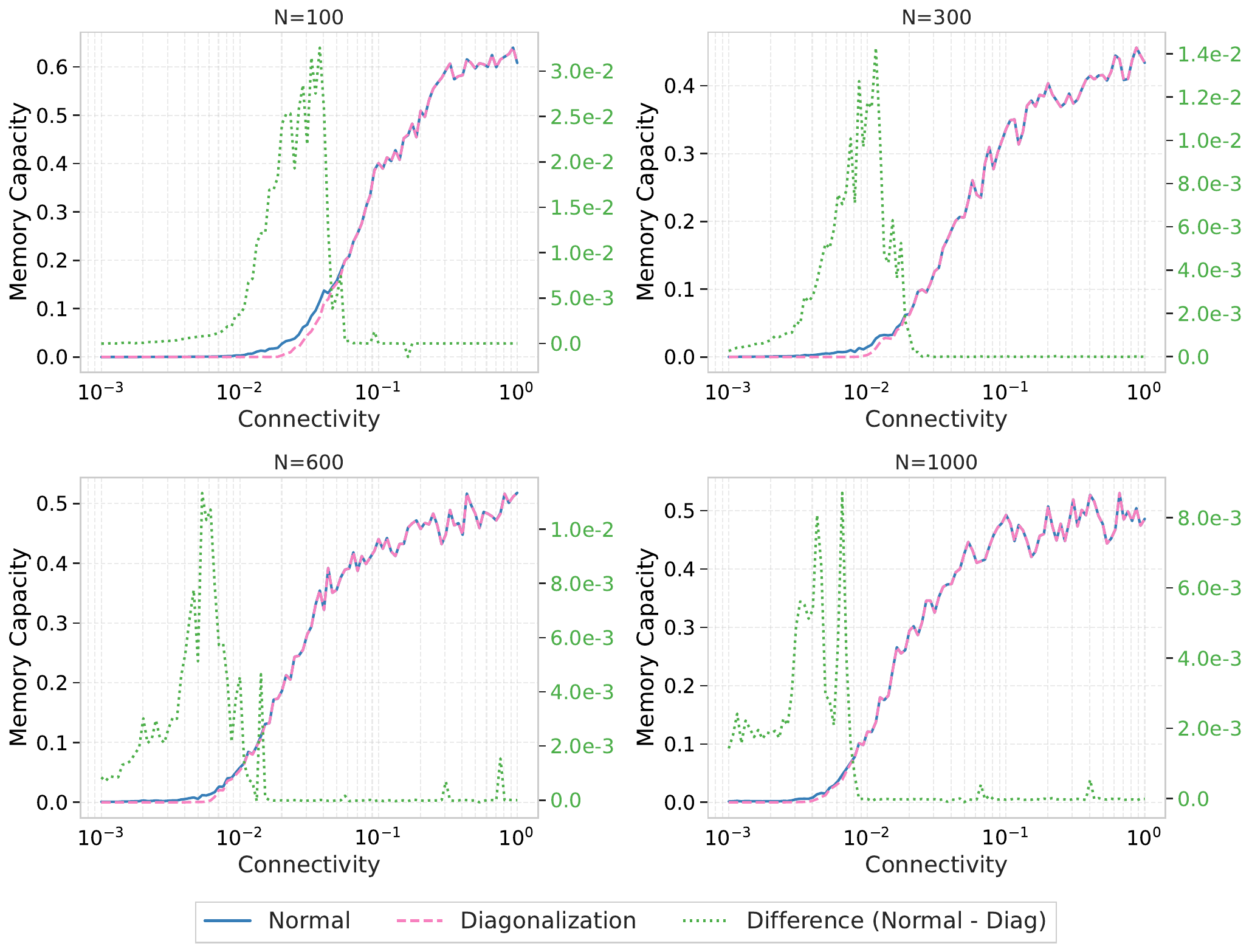}
    \caption{Memory capacity as a function of reservoir connectivity. The four subplots illustrate the memory capacity evaluated across varying levels of matrix connectivity for reservoir sizes $N=100, 300, 600, \text{and } 1000$. To ensure a task of moderate and proportional difficulty, the requested delay for each configuration is chosen so that the memory capacity yields 0.5 when the connectivity is one (to chose the corresponding value we looked at Figure \ref{fig:mc_delay}). The graphs compare the performance of the standard weight matrix approach (\textit{Normal}) against its diagonalized counterpart (\textit{Diagonalization}, corresponding to the EET and EWT methods), with the secondary axis tracking their absolute performance gap (\textit{Difference}). As expected, both methods experience a severe drop in memory capacity at extremely low connectivity levels; as the matrix $W$ becomes very sparse, the network loses the dimensionality required to maintain temporal representations. Moreover, the results highlight a specific minimum connectivity threshold which depends of the reservoir size. Below this connectivity, the \textit{Diagonalization} method begins to systematically underperform the \textit{Normal} baseline. At such extreme sparsity, the standard matrix $W$ lacks sufficient structural complexity for an effective eigendecomposition, causing the eigenspectrum to collapse and lose essential dimensional richness. Above this minimal threshold, the diagonalized approach successfully matches the standard baseline's performance.}
\end{figure}

\section{Conclusion}


Recent work in reservoir computing has highlighted two complementary trends that motivate the present study. On the one hand, several contributions have stressed the need to reduce the computational cost of reservoir-based models, either by simplifying the dynamics or by replacing recurrent computations with alternative feature constructions \citep{gauthier2021nextgenRC}. On the other hand, a number of theoretical and empirical studies have shown that linear reservoirs already possess strong dynamical and memory-related properties, and can constitute competitive building blocks for temporal processing \citep{dambre2012information, gonon2019reservoir, voelker2019legendre, tino2020dynamical, verzelli2021input}. In this context, our objective is not to revisit the modeling capabilities of linear reservoirs, but to revisit their implementation, and to investigate how their recurrent dynamics can be preserved while being expressed in a computationally more suitable form.

In this study, we showed that the dynamics of Linear Echo State Networks can be reformulated through matrix diagonalization, turning costly recurrent updates into efficient pointwise operations. This diagonalization-based optimization preserves the predictive power of standard linear ESNs while reducing computational complexity from $\mathcal{O}(N^2)$ to $\mathcal{O}(N)$. Beyond the theoretical speedup, our empirical results demonstrate that both the Eigenbasis Weight Transformation (EWT) and End-to-End Eigenbasis Training (EET) retain numerical stability and accuracy, with negligible differences from standard linear ESN outputs. We further introduced the Direct Parameter Generation (DPG) approach, which bypasses diagonalization entirely by directly sampling eigenvalues and eigenvectors, allowing better control over the eigenvalues that controls the dynamics of the reservoir, while replicating or even outperforming the behavior of classical Linear ESNs. 

Our investigation into the diagonalization of Linear ESNs reveals that the behavior of these models is almost entirely governed by their spectral properties. The success of the Direct Parameter Generation (DPG) approach, which bypasses the generation of an explicit weight matrix $W$ to directly sample eigenvalues, suggests that the explicit matrix structure of the reservoir is secondary to the distribution of its eigenvalues. This finding aligns with the observation that structured spectral coverages, such as the deterministic "Golden" distribution, can yield performance superior to purely stochastic initializations, particularly for tasks requiring long-term memory, though its superiority for tasks with different dynamic requirements remains an open question.

This spectral perspective is further reinforced by analyzing the trained readout weights in the eigenbasis. As illustrated in Figure~\ref{fig:eigenvalues_wout}, we observe a strong heterogeneity in the contribution of spectral components: only a subset of eigenvalues is associated with large output weights, indicating that the readout mechanism naturally selects the specific frequency components or decay rates relevant to the task. This fundamental reliance on the spectrum is mathematically corroborated by our demonstration that the input weights $W_{\text{in}}$ can be applied entirely after the recurrent updates. Since the temporal states depend strictly on the eigenvalues $\Lambda$, it proves that the memory and dynamical richness of the network are governed by its spectral configuration. Consequently, DPG positions itself as a rational approach to reservoir initialization rather than a computational optimization. It allows for the explicit selection of spectral components, shifting the paradigm from arbitrary matrix generation to controlled spectral initialization.

While our method offers significant advantages, it is important to note that the diagonalization of a pre-existing matrix (EWT) relies on the assumption that $W$ is diagonalizable and incurs a one-time $\mathcal{O}(N^3)$ preprocessing cost. Although this cost is amortized over long sequences, the DPG framework answers this bottleneck by avoiding matrix diagonalization entirely, making it highly scalable.
Finally, the independent evolution of state variables in the eigenbasis has implications beyond sequential efficiency. Since each neuron evolves independently, our formulation allows for the parallelization of state updates across neurons. 
Furthermore, because the dynamics are strictly linear, temporal computations can be parallelized across the entire sequence length using parallel scan algorithms \citep{blelloch1990prefix}, a property recently leveraged by modern state-space models like Mamba \citep{gu2024mamba}, suggesting that diagonalized Linear ESNs could leverage similar hardware optimizations. For future application, it could be interesting to adapt the methods with non linear readout \citep{gonon2019reservoir}.

These contributions suggest a more efficient use of linear reservoirs avoiding unnecessary computations, and a shift of paradigm in linear reservoir initialization towards the direct selection of spectral distributions.
In future works, it would be interesting to explore the various possible distributions for the eigenvalues, in order to systematically investigate their impact on reservoir dynamics and model performance.


\newpage
\appendix

\section{Practical Implementation: Memory Reinterpretation}
\label{app:implementation}

To maximize computational efficiency, we implement a hybrid approach: the reservoir update is performed by reinterpreting part of memory as complex for element-wise operations. Conversely, the training and readout steps remain strictly in the real domain to avoid supplementary complexity.

\subsection{Basis Construction}

Assume the reservoir matrix $W \in \mathbb{R}^{N \times N}$ is diagonalizable. Let $P \in \text{GL}_N(\mathbb{C})$ be the matrix of eigenvectors and $\Lambda$ the diagonal matrix of eigenvalues, partitioned as follows:
\begin{itemize}
    \item $\Lambda = \text{diag}(\lambda_1, \dots, \lambda_{N_r}, \mu_1, \bar{\mu}_1, \dots, \mu_{N_i}, \bar{\mu}_{N_i})$, where $\lambda_j \in \mathbb{R}$ and $\mu_k \in \mathbb{C} \setminus \mathbb{R}$.
    \item $P = [u_1, \dots, u_{N_r}, v_1, \bar{v}_1, \dots, v_{N_i}, \bar{v}_{N_i}]$ is the associated matrix of eigenvectors.
    \item $\Lambda_\text{real} = (\lambda_1, \dots, \lambda_{N_r})$ and $\Lambda_{cpx} = (\mu_1, \dots, \mu_{N_i})$.
\end{itemize}

\noindent Let define $Q$ as:
$$ Q = [u_1, \dots, u_{N_r}, \RE\{v_1\}, \IM\{v_1\}, \dots, \RE\{v_{N_i}\}, \IM\{v_{N_i}\}]$$

\noindent Applying Theorem~\ref{theorem_1} with the change of basis matrix $Q$ yields the transformed system. However, in the $Q$-basis, the reservoir matrix $[W]_{_Q} = Q^{-1}WQ$ is no longer diagonal; it is block-diagonal. Performing the update $[r(t-1)]_{_Q} [W]_{_Q}$ would typically require $2 \times 2$ matrix multiplications for the complex components, losing the efficiency of the diagonal representation.

\subsection{The Memory View Trick}

To bypass the computational bottleneck of block-diagonal matrix multiplications, we propose an efficient implementation based on memory reinterpretation. The core idea is to process the state update in the complex domain where $[W]_{_P}$ is purely diagonal, but to project it back into the real $Q$-basis for the readout phase. Since modern numerical libraries store complex numbers as pairs of contiguous real floating-point values, we treats the exact same physical memory block as either real or complex depending on the operation context. This allows us to perform the complex part of the reservoir update in the "complex world" and the real part of the reservoir update in the "real world" with fast element-wise operations ($\odot$), while keeping the readout training only in the "real world." While processing the state $r$ entirely in complex arithmetic would normally double the number of parameters required for the readout (due to separate real and imaginary weights), our approach guarantees that the readout cost remains perfectly equivalent to that of a standard real-valued ESN.

We define $[r]_{_Q}^{\text{real}}$ and $[r]_{_Q}^{\text{cpx}}$ as two different slices of the memory block $[r]_{_Q}$, with $[r]_{_Q}^{\text{cpx}}$ interpreted either as a list of couples reals number, or a list of complex number (containing real and imaginary part):
\begin{itemize}
    \item $[r]_{_Q}^{\text{real}} = [r]_{_Q}[1:N_r]$
    \item $[r]_{_Q}^{\text{cpx}} = [r]_{_Q}[N_r+1:N].\text{view}(\mathbb{C})$
\end{itemize}

\noindent Using this dual representation, the reservoir state update can be computed efficiently through separate element-wise multiplications for the real and complex partitions, before accumulating the input and feedback terms in the real domain. \\

\noindent \textbf{Reservoir Update Step:}
\begin{align*}
    [r]_{_Q}^{\text{real}} &\gets [r]_{_Q}^{\text{real}} \odot \Lambda_{\text{real}} \\
    [r]_{_Q}^{\text{cpx}} &\gets [r]_{_Q}^{\text{cpx}} \odot \Lambda_{\text{cpx}} \\
    [r]_{_Q} &\gets [r]_{_Q} + u(t)[W_{\text{in}}]_{_Q} + y(t-1)[W_{\text{fb}}]_{_Q}
\end{align*}

\begin{proof}

We define the transformation matrix $\mathcal{Z} = \text{diag}(I_{n_r}, Z, \dots, Z)$, where $Z$ is the change-of-basis matrix mapping a complex conjugate pair to its real and imaginary parts:
\[ Z = \frac{1}{2} \begin{bmatrix} 1 & 1 \\ -i & i \end{bmatrix} \]
Then, $Q = P\mathcal{Z}$. By definition of the transformation bases, the relation $[r]_{_Q} = r Q = r P \mathcal{Z} = [r]_{_P} \mathcal{Z}$ always holds. We can expand this explicitely:

\begin{equation*}
[r]_{_Q} = \begin{bmatrix} [r]_{_P}^{(\lambda_1)} & \dots & [r]_{_P}^{(\lambda_{n_r})} & \RE\{[r]_{_P}^{(\mu_1)}\} & \IM\{[r]_{_P}^{(\mu_1)}\} & \dots & \RE\{[r]_{_P}^{(\mu_{n_i})}\} & \IM\{[r]_{_P}^{(\mu_{n_i})}\} \end{bmatrix}
\end{equation*}
\noindent When applying the memory views, the data partitions correspond exactly to the real eigenvalues and the complex eigenvalues in the original $P$-basis:
\begin{align*}
[r]_{_Q}^\text{real} \stackrel{\text{def}}{=} ([r]_{_Q})[1:n_r] &= \begin{bmatrix} [r]_{_P}^{(\lambda_1)} & \dots & [r]_{_P}^{(\lambda_{n_r})} \end{bmatrix} \\
[r]_{_Q}^\text{cpx} \stackrel{\text{def}}{=} ([r]_{_Q})[n_r+1:n].\text{view}(\mathbb{C}) &= \begin{bmatrix} [r]_{_P}^{(\mu_1)} & \dots & [r]_{_P}^{(\mu_{n_i})} \end{bmatrix}
\end{align*}

\noindent Because the memory perfectly aligns with the real and imaginary components of the $P$-basis projection, performing the element-wise multiplications $[r]_{_Q}^{\text{real}} \gets [r]_{_Q}^{\text{real}} \odot \Lambda_{\text{real}}$ and $[r]_{_Q}^{\text{cpx}} \gets [r]_{_Q}^{\text{cpx}} \odot \Lambda_{\text{cpx}}$ on the specific views is strictly equivalent to performing the full complex element-wise multiplication in the original $P$-basis:
\begin{equation}
    [r]_{_P} \gets [r]_{_P} \odot \Lambda
\end{equation}

\noindent By reverting to the matrix multiplication notation, this operation directly corresponds to the state transition in the $Q$-basis:
\begin{equation}
    [r]_{_Q} \gets [r]_{_Q} [W]_{_Q}
\end{equation}

\noindent Finally, by accumulating the input real value and feedback projections directly onto the real memory block $[r]_{_Q}$, the complete reservoir step is correctly and efficiently evaluated as:
\begin{equation}
    [r]_{_Q} \gets [r]_{_Q} [W]_{_Q} + u(t)[W_{\text{in}}]_{_Q} + y(t-1)[W_{\text{fb}}]_{_Q}
\end{equation}
\end{proof}

\subsection{Readout and Training}

While the reservoir step partly exploits complex arithmetic to achieve optimal parallel efficiency, the readout and training are performed entirely in the $Q$-basis. Following Theorem~\ref{theorem_1} (iii) and (iv), the readout weights are learned via:

\begin{align}
    y(t) &= W_\text{out,bias} + y(t-1)\, W_\text{out,out} + [r(t)]_{_Q}\, [W_\text{out,res}]_{_Q}
\end{align}

\begin{equation}
    [W_{\text{out}}]_{_Q} = \left( [X]_{_Q}^T [X]_{_Q} + \alpha \begin{bmatrix} I & 0 \\ 0 & Q^T Q \end{bmatrix} \right)^{-1} [X]_{_Q}^T Y
\end{equation}






\section{Special case of $D_\text{in} = 1$}
\label{app:special_case_win}

This section follows the reasoning present in section \ref{sec:win_end}. Suppose now that the dimension of the input signal is $D_\text{in} = 1$ and $r(0)=0$. Then $R(t)$ is effectively a row vector $1 \times N$.

\begin{align*}
    y(t) &= r(t) W_\text{out} \\
    &= (w_\text{in} \odot R(t)) W_\text{out} \\
    &= R(t) \left( \begin{bmatrix} w_\text{in}^T \dots w_\text{in}^T \end{bmatrix} \odot W_\text{out} \right)
\end{align*}

\noindent Let $\mathcal{L}$ be the objective loss function (e.g., Mean Squared Error) and $Y(t)$ the truth signal. The standard readout $W_\text{out}$ is typically trained such that:
\begin{equation*}
    W_\text{out} = \argmin_{\beta \in \mathbb{C}^{N \times D_\text{out}}} \mathcal{L}\left( \left( r(t)\beta, Y(t) \right)_{1 \leq t \leq T} \right)
\end{equation*}

\noindent \textbf{Note:} If an L2 penalty (Ridge regression) is applied to regularize the readout, changing the parameterization alters the problem. Penalizing the transformed composite weights instead of $\|W_\text{out}\|^2$ will yield a different effective regularization, meaning standard Ridge regression is not perfectly equivalent under this transformation. However, for unregularized optimization, the following holds:

\begin{theorem}
If no weight in $w_\text{in}$ is 0, then the optimization can be performed directly on the unweighted states $R(t)$:
\begin{equation*}
    \gamma^* = \argmin_{\gamma \in \mathbb{C}^{N \times 1}} \mathcal{L}\left( \left( R(t)\gamma, Y(t) \right)_{1 \leq t \leq T} \right)
\end{equation*}
Once the optimal $\gamma^*$ is found, the original output weights can be uniquely recovered via $W_\text{out} = \gamma^* \oslash \begin{bmatrix} w_\text{in}^T \dots w_\text{in}^T \end{bmatrix}$, where $\oslash$ denotes element-wise division.
\end{theorem}

\begin{proof}
Let $\gamma = \begin{bmatrix} w_\text{in}^T \dots w_\text{in}^T \end{bmatrix} \odot W_\text{out}$. Since $w_\text{in}$ contains no zeros, there is a bijection between $W_\text{out}$ and $\gamma$. The output equation becomes:
\begin{align*}
    y(t) &= R(t) (\begin{bmatrix} w_\text{in}^T \dots w_\text{in}^T \end{bmatrix} \odot W_\text{out}) \\
         &= R(t) \gamma
\end{align*}
Thus, we can optimize for $\gamma$ directly using the states $R(t)$, which do not depend on $w_\text{in}$.
\end{proof}

\noindent Under these hypotheses, the effective readout weights can be learned without explicitly instantiating $w_\text{in}$ during the state collection phase.

\acks{Experiments presented in this paper were done using the Plafrim experimental testbed, supported by Inria, CNRS (LABRI and IMB), Université de Bordeaux, Bordeaux INP and Conseil Régional d’Aquitaine. We would like to thank the Plafrim team for providing computational resources and technical support.}

\newpage

\vskip 0.2in
\bibliography{references}

\end{document}